\documentclass{amsart}
\newtheorem{theorem}{Theorem}[section]

\theoremstyle{definition}

\newtheorem{definition}{Definition}
\newtheorem{corollary}{Corollary}

\newtheorem{example}{Example}
\theoremstyle{remark}
\newtheorem{remark}[theorem]{Remark}
\usepackage{ragged2e}
\usepackage{amsmath}
\numberwithin{equation}{section}
\usepackage[table, svgnames, dvipsnames]{xcolor}
\usepackage{makecell, cellspace, caption}
\usepackage{tabularx, ragged2e, booktabs, caption}
\usepackage{amssymb}
\usepackage{amsmath}
\usepackage{color,soul}
\usepackage{enumitem}
\usepackage{tabularx,ragged2e,booktabs,caption}
\usepackage{longtable}
\usepackage{graphicx}
\usepackage[colorlinks]{hyperref}
\DeclareMathOperator{\rank}{rank}
\DeclareMathOperator{\SHull}{SHull}
\DeclareMathOperator{\Hull}{Hull}
\DeclareMathOperator{\LSHull}{LSHull}
\DeclareMathOperator{\RSHull}{RSHull}
\begin{document}
	\title{Symplectic Hulls over a Non-Unital Ring}
	\author{Anup Kushwaha}
    \address{Department of Mathematics, Indian Institute of Technology Patna, Patna-801106}
    \curraddr{}
    \email{E-mail: anup$\textunderscore$2221ma11@iitp.ac.in}
    \thanks{}

   \author{ Om Prakash$^*$}
   \address{Department of Mathematics, Indian Institute of Technology Patna, Patna-801106}
   \curraddr{}
   \email{om@iitp.ac.in}
   \thanks{* Corresponding author}
	
	
	
	\dedicatory{}

\begin{abstract}
This paper presents the study of the symplectic hulls over a non-unital ring $ E= \langle \kappa,\tau \mid 2 \kappa =2 \tau=0,~ \kappa^2=\kappa,~ \tau^2=\tau,~ \kappa \tau=\kappa,~ \tau \kappa=\tau \rangle$. We first identify the residue and torsion codes of the left, right, and two-sided symplectic hulls, and characterize the generator matrix of the two-sided symplectic hull of a free $E$-linear code. Then, we explore the symplectic hull of the sum of two free $E$-linear codes. Subsequently, we provide two build-up techniques that extend a free $E$-linear code of smaller length and symplectic hull-rank to one of larger length and symplectic hull-rank. Further, for free $E$-linear codes, we discuss the permutation equivalence and investigate the symplectic hull-variation problem. An application of this study is given by classifying the free $E$-linear optimal codes for smaller lengths.

\end{abstract}

\maketitle

\section{Introduction}

 In 1990, Assmus and Key \cite{Assmus} introduced the hull of a code while analyzing some properties of finite projective planes. Let $C$ be a linear code of length $n$ over the finite field $\mathbb{F}_q$ with $q$ elements. The hull of $C$ is the intersection of $C$ and its dual $C^\perp$. It has significant importance in both coding theory and cryptography, as many algorithms for computing automorphism groups and checking permutation equivalence rely critically on the dimension of the hull \cite{Leon1,Leon2,Send1997,Send2000}. The computational complexity of these algorithms reduces significantly when the hull dimension is small. Moreover, the hull plays a crucial role in the construction of quantum codes \cite{Dougherty,Liu2}. In view of these diverse applications, the hulls and their associated properties have been studied substantially. \par
Recall that the Euclidean dual of an $\mathbb{F}_q$-linear code $C$ of length $n$ is denoted by $C^\perp$, and defined as the collection of all elements of $\mathbb{F}_q^n$ that are orthogonal to $C$ under Euclidean inner product. Further, $\Hull(C)=C \cap C^\perp$, and if it is trivial, the code $C$ is called a linear complementary dual (LCD) code. This family of codes has applications in secret sharing schemes \cite{Yadav2025} as well as in cryptographic systems where they offer resistance to side-channel attacks \cite{Carl16}. In 2004, the author in \cite{Send04} employed the notion of hull dimension spectra to show that LCD codes achieve the Gilbert–Varshamov bound. Moreover, Islam and Prakash \cite{Islam2022} investigated cyclic codes over non-chain rings and constructed quantum and LCD codes, while double circulant LCD codes were explored in \cite{Prakash}. \\
On the other hand, the symplectic inner product of two vectors $x=(u|v)$ and $y=(w|z)$ of $\mathbb{F}_q^{2n}$ is defined by 
$$\langle x,y \rangle_s=\langle u, z \rangle- \langle v, w \rangle,$$
where $u,v,w,z \in \mathbb{F}_q^n$ and $\langle, \rangle$ denote the Euclidean inner product. While plenty of works on Euclidean hulls are available \cite{Debnath,Dougherty2,Li,Qian,Sok,Wang}, there are a very few works that dealt with symplectic hulls \cite{Li2024,Xu}.\par

In fact, the following research gaps motivate us to present the current paper on symplectic hulls over a non-unital ring $ E= \langle \kappa,\tau \mid 2 \kappa =2 \tau=0,~ \kappa^2=\kappa,~ \tau^2=\tau,~ \kappa \tau=\kappa,~ \tau \kappa=\tau \rangle$:
\begin{itemize}
    \item Over the past five years, researchers have shown a growing interest in exploring various families of codes over non-unital rings \cite{Alah23,Deb,Kushwaha1,Manseri2}. In 2022, Shi et al. \cite{Shi21} initiated the study of Euclidean LCD codes over a non-unitary ring. Motivated by their research, Manseri and Solé \cite{Manseri1} were the first to investigate the symplectic LCD codes over a non-unital ring. Recently, Li and Zhu \cite{Li2024} explored symplectic hulls over finite fields. In 2026, we studied the Euclidean hulls over the non-unital ring $E$ \cite{Kushwaha3}. However, any type of study on symplectic hulls over a non-unital ring is not present in the literature. This research gap is our primary motivation to consider the symplectic inner product to study the hulls over the non-unital ring $E$. We also study the two-sided symplectic hull of the sum of two free $E$-linear codes.
    
    \item The build-up construction is a well-established technique that extends a linear code with a smaller length to one of larger length. Using this technique, Type IV and QSD codes have been classified over non-unital rings for smaller lengths in \cite{Alah22,Kim}. Subsequently, build-up construction techniques were employed in \cite{Li} for classifying binary linear codes with a fixed Euclidean hull dimension. However, there is no work available in the literature that dealt with the build-up techniques for symplectic hulls over a finite field or a finite ring (with or without unity). This research gap motivated us to propose two build-up construction techniques over the non-unital ring $E$ that extend a free $E$-linear code with a smaller length and symplectic hull-rank to one of larger length and symplectic hull-rank.
   
    \item The length and dimension of permutation-equivalent linear codes are always identical \cite{Huffman}, but they need not be equal in their duals \cite{Chen}. Consequently, Chen \cite{Chen} proposed and studied the Euclidean hull-variation problem for linear codes over finite fields. Then, Li and Zhu \cite{Li2024} investigated the symplectic hull-variation problem for $\mathbb{F}_q$-linear codes. Motivated by these, we study the symplectic hull-variation problem for the free $E$-linear codes.
\end{itemize}
We structure this work as follows: Section $2$ deals with the basic concepts of codes over the non-unital ring $E$, while Section $3$ investigates the various symplectic hulls over the ring $E$ and identifies their residue and torsion codes. Additionally, this Section examines the two-sided symplectic hull of the sum of two free $E$-linear codes. Section $4$  provides two build-up construction techniques over the ring $E$ and also gives some illustrative examples. Section $5$ discusses the permutation equivalence and investigates the symplectic hull-variation problem for free $E$-linear codes and also classifies free $E$-linear optimal codes for lengths up to $4$. Section $6$ concludes the work with some possible future directions.
\section{Preliminaries}
The noncommutative ring $ E= \langle \kappa,\tau \mid 2 \kappa =2 \tau=0,~ \kappa^2=\kappa,~ \tau^2=\tau,~ \kappa \tau=\kappa,~ \tau \kappa=\tau \rangle$ appeared in the classification of Fine \cite{Fine93}, is a non-unital ring with $4$ elements. In this work, we consider the ring $E$ as a code alphabet and study the symplectic hulls of $E$-linear codes. The ring $E$ is a finite ring of characteristic $2$ and contains $4$ elements $\{ ~i \kappa+j \tau~ |~ 0 \leq i,j <2 ~\}$. If we denote $\kappa + \tau=\zeta$, then $E=\{0,\kappa, \tau, \zeta\}$. To characterize the algebraic structure of the ring $E$, its multiplication table is given in Table \ref{Tab2}.

\begin{table}[ht]
\centering
\caption{\label{Tab2} Multiplication table.}
\begin{tabular}{|c|c|c|c|c|}
\hline
  $\cdot$ & $0$ & $\kappa$ & $\tau$ &  $\zeta$ \\
 \hline
 $0$ & $0$ & $0$ & $0$ & $0$  \\

\hline
 $\kappa$ & $0$ & $\kappa$ & $\kappa$ &  $0$ \\
 \hline
 $\tau$ & $0$ & $\tau$ & $\tau$ & $0$
 \\
 
 \hline
 $\zeta$ & $0$ & $\zeta$ & $\zeta$ & $0$
 \\
 
\hline
\end{tabular}
\end{table}

From the multiplication table, it follows that the ring $E$ is non-unitary and possesses a unique maximal ideal $J=\{0,\zeta\}$. Consequently, $E$ is a local ring, and its residue field is $E/J=\mathbb{F}_2$. In addition, $e\kappa = e\tau = e$ for all $e \in E$. Moreover, every element $e \in E$ admits a $\zeta$-adic decomposition of the form
 $$e=u \kappa+v \zeta ~~ \text{where}~ u,v \in \mathbb{F}_2.$$
 For all $e \in E$ and $v \in \mathbb{F}_2$, we can define an action of $\mathbb{F}_2$ on $E$ as $ev=ve$. This action is distributive, i.e., $e(u\oplus_{\mathbb{F}_2} v)=ue+ve$ for all $ e \in E$ and $u,v \in \mathbb{F}_2$. Now, the map of reduction modulo $J$ is a  map $\pi:E\rightarrow E/J =\mathbb{F}_2$ defined by
 $$\pi(e)=\pi(u\kappa+v\zeta)=u.$$
 This map extends naturally (componentwise) to a map from $E^{2n}$ to  $\mathbb{F}_2^{2n}$. \par

\begin{definition} A left $E$-submodule of $E^{2n}$ refers to an $E$-linear code of length $2n$.
\end{definition}
For an $E$-linear code $C$, we associate with $C$ the following two $\mathbb{F}_2$-linear codes.
\begin{enumerate}
\item \textbf{Residue code:} The residue code of $C$, denoted by $C_{\mathrm{Res}}$, is defined by 
    $$ C_{Res}=\{\pi(z)~ |~ z \in C \}.$$
    \item \textbf{Torsion code:} The $\mathbb{F}_2$-linear code defined by
    $$ C_{Tor}=\{ v \in \mathbb{F}_{2}^{2n} ~ |~ v \zeta \in C\},$$
    is called the torsion code of an $E$-linear code $C$ of length $2n$.
    \end{enumerate}

 \begin{remark}
        Throughout this paper, we fix the notations $\kappa$ and $\tau$ for the generators of the ring $E$, and $ \zeta = \kappa+ \tau $. Moreover, $k_1=\dim(C_{Res})$ and $k_2=\dim(C_{Tor})-k_1$. 

     \end{remark} 
 If $C$ is an $E$-linear code and $S=\{  \alpha_1,\alpha_2,\ldots , \alpha_k \} \subset C$, then the $E$-span of $S$ is given by $$\langle S \rangle_{E} = \{ e_1\alpha_1+e_2\alpha_2+ \cdots + e_k\alpha_k ~|~ e_j \in E, ~  1 \leq j \leq k \},$$
and the additive span of the set $S$ is given by
$$\langle S \rangle_{\mathbb{F}_2} = \{ u_1\alpha_1+u_2\alpha_2+ \cdots + u_k\alpha_k ~|~ u_j \in \mathbb{F}_2, ~ 1 \leq j \leq k  \}.$$
 As the ring $E$ lacks a multiplicative identity, the additive span of the set $S$ is not necessarily a subset of its $E$-span. The set $S=\{  \alpha_1,\alpha_2,\ldots , \alpha_k \}$ is termed a generating set for the code $C$ if $C$ is the union of both the spans defined above, i.e.,
    $$\langle S \rangle_{E} \cup \langle S \rangle_{\mathbb{F}_2}=C.$$ 
Now, if $C$ is an $E$-linear code  of length $2n$ and $S=\{  \alpha_1,\alpha_2,\ldots , \alpha_k \}\subset C$ is its generating set, then a $k \times 2n$ matrix $G$ with rows $\alpha_1,\alpha_2,\ldots , \alpha_k$ satisfying $C=\langle S \rangle_{E} \cup \langle S \rangle_{\mathbb{F}_2}$ refers to a generator matrix  of $C$. A parity-check matrix of $C$ refers to a generator matrix of $C^{\perp_S}$. The freeness of an $E$-linear code $C$ is characterized by the condition $C_{Res} = C_{Tor}$. Equivalently, $C$ is free if $k_2=0$. Then the rank of a free $E$-code $C$, denoted by $\rank(C)$, is given by the cardinality of a minimal generating set of $C$. The symplectic weight of a codeword $w=(w_1,w_2, \ldots, w_{2n})$ of an $E$-linear code $C$ of length $2n$ is defined as 
$$wt_s(w)=|\{~j ~|~ (w_j,w_{n+j}) \neq (0,0), 1 \leq j \leq n\}|.$$ Then the quantity 
$$d_s(C)=\text{min}\{wt_s(w)~|~ w \in C \setminus \{\boldsymbol{0}\} \}$$
 is called the symplectic distance of the code $C$.

 The Euclidean inner product of any vectors  $w=(w_1, w_2, \ldots, w_n)$ and $z=(z_1,z_2,\ldots,z_n)$ of $E^n$ is defined as
    $ \langle w, z  \rangle=\sum_{j=1}^{n}w_j z_j$. Subsequently, for any vectors $x=(u | v)$ and $y=(u' | v')$ of $E^{2n}$, we define
    $$ \langle x, y \rangle_s=\langle u, v' \rangle + \langle v, u' \rangle$$
    to be their symplectic inner product.
  If $\Omega_{2n}=\begin{pmatrix}
      \textbf{0}_n & I_n \\
      I_n & \textbf{0}_n
  \end{pmatrix}$ where $I_n$ is the identity matrix of order $n$, then the symplectic inner product of vectors $x,y \in E^{2n}$ can also be defined as $\langle x,y \rangle_s= x \Omega_{2n}y^T $. Here, $y^T$ denotes the transpose of $y$. Consequently, the left and right symplectic duals of an $E$-code $C$, denoted by $C^{\perp_{S_L}}$ and $C^{\perp_{S_L}}$, respectively, are the $E$-linear codes given by
    $$C^{\perp_{S_L}}= \{ \boldsymbol{z} \in E^{2n}~ |~\langle \boldsymbol{z}, \boldsymbol{w}  \rangle_s= 0,\forall ~\boldsymbol{w}\in C\},$$
    and 
        $$C^{\perp_{S_R}}= \{ \boldsymbol{z }\in E^{2n}~ |~\langle \boldsymbol{w}, \boldsymbol{z}  \rangle_s= 0,\forall ~\boldsymbol{w} \in C\}.$$
        The intersection of the above two symplectic duals is called the two-sided symplectic dual, and is denoted by $C^{\perp_s}$, i.e., $C^{\perp_s}=C^{\perp_{S_L}} \cap C^{\perp_{S_R}}$. Further, an $E$-linear code of length $2n$ is termed as left, right and two-sided symplectic nice if $|C||C^{\perp_{S_L}}|=|E|^{2n}$, $|C||C^{\perp_{S_R}}|=|E|^{2n}$ and $|C||C^{\perp_s}|=|E|^{2n}$, respectively. In addition, if $C$ satisfies $C \cap C^{\perp_s}=\{0\}$, it is called a symplectic LCD code. Further, the left, right and two-sided symplectic hulls of an $E$-linear code $C$ are defined respectively as
$$\LSHull(C)=C\cap C^{\perp_{S_L}},~~  \RSHull(C)=C\cap C^{\perp_{S_R}}, ~~\text{and}~~  \SHull(C)=C\cap C^{\perp_S}.$$
Throughout, a free $E$-linear code of length $2n$ and rank $k$ is denoted by $[2n,k]$.

\section{Symplectic hulls}

This section focuses on the investigation of three symplectic hulls of an $E$-linear code, namely, left, right and two-sided symplectic hulls. The residue and torsion codes of these symplectic hulls are calculated, and their basic properties are discussed. Further, we characterize the form of a generator matrix of the two-sided symplectic hull of a free $E$-linear code and compute the corresponding symplectic hull-rank. It also investigates the symplectic hull of the sum of two free $E$-linear codes.\par

 The following key result corresponds to Lemma $1$ and Theorem $7$ of \cite{Alah23}.

\begin{theorem}\cite{Alah23}\label{Thm1c} 
    For an  $E$-linear code $C$, $\kappa C_{Res} \subseteq C$ and $\zeta C_{Tor} \subseteq C$. Moreover, $C=\kappa C_{Res} \oplus \zeta C_{Tor}$.
\end{theorem}

Here, we calculate the residue and torsion codes of the three symplectic dual codes.

\begin{theorem}\label{Thm1b}
 For an  $E$-linear code $C$ of length $2n$,
    
\begin{enumerate}
    
    \item $(C^{\perp_{S_L}})_{Res}=(C_{Res})^{\perp_S}=(C^{\perp_{S_L}})_{Tor}$;
    \item $(C^{\perp_{S_R}})_{Res}=(C_{Tor})^{\perp_S}$ and $(C^{\perp_{S_R}})_{Tor}=\mathbb{F}_2^{2n}$;
    \item $(C^{\perp_S})_{Res}=(C_{Tor})^{\perp_S}$ and $(C^{\perp_S})_{Tor}=(C_{Res})^{\perp_S}$.
    
    \end{enumerate}
\end{theorem}
\begin{proof}
    \begin{enumerate}
        \item This statement is precisely Lemma $7$ of \cite{Manseri1}.
        \item Let $ u \in (C^{\perp_{S_R}})_{Res}$ and $v \in C_{Tor}$. By Theorem \ref{Thm1c}, $\kappa u \in C^{\perp_{S_R}}\in $ and $\zeta v \in C$. Then, the inner product $\langle \zeta v, \kappa u \rangle_s=\zeta \langle v,u \rangle_s=0$ implies that $\langle v,u \rangle_s=0$. Hence, $ u \in (C_{Tor})^{\perp_S}$, and so $(C^{\perp_{S_R}})_{Res} \subseteq (C_{Tor})^{\perp_S}$. For the converse, let $v \in (C_{Tor})^{\perp_S}$. Also, suppose that $x \in C$. From Theorem \ref{Thm1c}, there exist some $u \in C_{Res}$ and $w \in C_{Tor}$ such that $x= \kappa u + \zeta w$. Then, we have 
        $$\langle x, \kappa v \rangle_s=\langle \kappa u + \zeta w, \kappa v \rangle_s= \kappa \langle u, v \rangle_s+ \zeta \langle w,v \rangle_s =0. $$
      This shows that $\kappa v \in C^{\perp_{S_R}}$. Since $\pi(\kappa v)=v$, $v \in (C^{\perp_{S_R}})_{Res}$. Therefore, $(C_{Tor})^{\perp_S} \subseteq (C^{\perp_{S_R}})_{Res}$. Thus, $(C^{\perp_{S_R}})_{Res}=(C_{Tor})^{\perp_S}$.\par
   In order to show $(C^{\perp_{S_R}})_{Tor}=\mathbb{F}_2^{2n}$, we only need to prove $\mathbb{F}_2^{2n} \subseteq (C^{\perp_{S_R}})_{Tor}$. For this, let $ u \in \mathbb{F}_2^{2n}$ and $x \in C$. By Theorem \ref{Thm1c}, there exist some $v \in C_{Res}$ and $w \in C_{Tor}$ such that $x= \kappa v + \zeta w$. Then, the inner product
    $$\langle x, \zeta u \rangle_s=\langle \kappa v + \zeta w, \zeta u \rangle_s=  \langle \kappa v, \zeta u \rangle_s+  \langle \zeta w,\zeta u \rangle_s =0. $$
Hence, $\zeta u \in C^{\perp_{S_R}}$, and so $ u \in (C^{\perp_{S_R}})_{Tor}$. Therefore, $\mathbb{F}_2^{2n} \subseteq (C^{\perp_{S_R}})_{Tor}$ and the proof is complete.

        \item By the first two parts and Theorem \ref{Thm1c}, we have 
     \begin{align*}
        C^{\perp_S} &=C^{\perp_{S_L}} \cap C^{\perp_{S_R}} \\
        &=(\kappa (C_{Res})^{\perp_S} \oplus \zeta (C_{Res})^{\perp_S}) \cap (\kappa (C_{Tor})^{\perp_S} \oplus \zeta \mathbb{F}_2^{2n}) \\
        &= \kappa ( (C_{Res})^{\perp_S} \cap (C_{Tor})^{\perp_S}) \oplus \zeta ((C_{Res})^{\perp_S} \cap \mathbb{F}_2^{2n}) \\
        &= \kappa (C_{Tor})^{\perp_S} \oplus \zeta (C_{Res})^{\perp_S}.
        \end{align*}
       Thus, $(C^{\perp_S})_{Res}=(C_{Tor})^{\perp_S}$ and $(C^{\perp_S})_{Tor}=(C_{Res})^{\perp_S}$.
    \end{enumerate}
\end{proof}

The following corollary is immediate by using Theorems \ref{Thm1c} and \ref{Thm1b}. 

\begin{corollary}\label{Cor3}
    Let $C$ be an $E$-linear code of length $2n$. Then
    \begin{enumerate}
        \item $C^{\perp_{S_L}}=\kappa (C_{Res})^{\perp_S} \oplus \zeta (C_{Res})^{\perp_S}$;
        \item $C^{\perp_{S_R}}=\kappa (C_{Tor})^{\perp_S} \oplus \zeta \mathbb{F}_2^{2n}$;
        \item $C^{\perp_S}=\kappa (C_{Tor})^{\perp_S} \oplus \zeta (C_{Res})^{\perp_S}$.
    \end{enumerate}
\end{corollary}

 Now, we check the freeness of the three symplectic duals of an $E$-linear code.

 \begin{theorem}
     If $C\neq \{\boldsymbol{0}\}$ is an $E$-linear code, then $C^{\perp_{S_L}}$ is always free but $C^{\perp_{S_R}}$ is never free. Further, $C^{\perp_S}$ is free if and only if $C$ is free.
 \end{theorem}
\begin{proof}
  By Theorem $3$ of \cite{Manseri1}, the freeness of the left symplectic dual and the non-freeness of the right symplectic dual are immediate. On the other hand, if $C$ is a free $E$-linear code, then from $C_{Res}=C_{Tor}$ and Theorem \ref{Thm1b}(3), we have 
  $$(C^{\perp_S})_{Res}=(C_{Tor})^{\perp_S}=(C_{Res})^{\perp_S}=(C^{\perp_S})_{Tor}.$$
  Therefore, the symplectic dual is also free. For the converse part, suppose that the symplectic dual $C^{\perp_S}$ is free. The freeness of $C^{\perp_S}$ and Theorem \ref{Thm1b}(3) imply that $(C_{Res})^{\perp_S}=(C_{Tor})^{\perp_S}$. Therefore, $C_{Res}=C_{Tor}$, and hence $C$ is free.
  
\end{proof}

The following result calculates the residue and torsion codes of the three symplectic hulls of an $E$-linear code.

\begin{theorem}\label{Thm2}
    For an $E$-linear code $C$,
\begin{enumerate}
    \item $(\LSHull(C))_{Res}=C_{Res}\cap (C_{Res})^{\perp_S}$ and $(\LSHull(C))_{Tor}=(C_{Res})^{\perp_S} \cap C_{Tor}$;
    \item $(\RSHull(C))_{Res}=C_{Res} \cap (C_{Tor})^{\perp_S}$ and $(\RSHull(C))_{Tor}=C_{Tor}$;
    \item $(\SHull(C))_{Res}=C_{Res} \cap (C_{Tor})^{\perp_S}$ and $(\SHull(C))_{Tor}=(C_{Res})^{\perp_S} \cap C_{Tor}$.
\end{enumerate}
\end{theorem}
\begin{proof}  \begin{enumerate}
   \item Given an $E$-linear code $C$, by Corollary \ref{Cor3}(1), we have $C^{\perp_{S_L}}=\kappa (C_{Res})^{\perp_S} \oplus \zeta (C_{Res})^{\perp_S}$. Therefore,
\begin{align*}
    \LSHull(C) &= C \cap C^{\perp_{S_L}} \\
         & = (\kappa C_{Res}\oplus \zeta C_{Tor})\cap (\kappa (C_{Res})^{\perp_S} \oplus \zeta (C_{Res})^{\perp_S}) \\
         &= \kappa (C_{Res}\cap (C_{Res})^{\perp_S}) \oplus \zeta (C_{Res})^{\perp_S} \cap C_{Tor}).
\end{align*}
   Thus, $(\LSHull(C))_{Res}=C_{Res}\cap (C_{Res})^{\perp_S}$ and $(\LSHull(C))_{Tor}=(C_{Res})^{\perp_S} \cap C_{Tor}$.

   \item For an $E$-linear code $C$ of length $2n$, Corollary \ref{Cor3}(2) implies that $C^{\perp_{S_R}}=\kappa (C_{Tor})^{\perp_S} \oplus \zeta \mathbb{F}_2^{2n}$. Then 
   \begin{align*}
      \RSHull(C) &= C \cap C^{\perp_{S_R}} \\
         & = (\kappa C_{Res}\oplus \zeta C_{Tor})\cap (\kappa (C_{Tor})^{\perp_S} \oplus \zeta \mathbb{F}_2^{2n}) \\
         &= \kappa (C_{Res} \cap (C_{Tor})^{\perp_S}) \oplus \zeta (  C_{Tor} \cap \mathbb{F}_2^{2n}) \\
         &= \kappa (C_{Res} \cap (C_{Tor})^{\perp_S}) \oplus \zeta (  C_{Tor}).
   \end{align*}
   Therefore, $(\RSHull(C))_{Res}=C_{Res} \cap (C_{Tor})^{\perp_S}$ and $(\RSHull(C))_{Tor}=C_{Tor}$.

   \item Again, from Corollary \ref{Cor3}(3), $C^{\perp_S}=\kappa (C_{Tor})^{\perp_S} \oplus \zeta (C_{Res})^{\perp_S}$. Therefore, 
      \begin{align*}
    \SHull(C)&= C \cap C^{\perp_S}  \\
    &=(\kappa C_{Res}\oplus \zeta C_{Tor})\cap (\kappa (C_{Tor})^{\perp_S} \oplus \zeta (C_{Res})^{\perp_S}) \\
    &= \kappa (C_{Res} \cap (C_{Tor})^{\perp_S}) \oplus \zeta ((C_{Res})^{\perp_S} \cap C_{Tor}).
    \end{align*}
   Thus, $(\SHull(C))_{Res}=C_{Res} \cap (C_{Tor})^{\perp_S}$ and $(\SHull(C))_{Tor}=(C_{Res})^{\perp_S} \cap C_{Tor}$.
  
\end{enumerate}
\end{proof} 

The above theorem immediately yields the following important corollary.
\begin{corollary}\label{Cor1}
  For an $E$-linear code  $C$,
   $$(\SHull(C))_{Res} \subseteq \SHull(C_{Res}) ~\text{but}~~ (\SHull(C))_{Tor} \supseteq \SHull(C_{Tor}).$$ 
Moreover, equality holds if $C$ is free.
\end{corollary}
\begin{proof}
  For an $E$-linear code $C$, Theorem \ref{Thm2}(3) provides 
   $$(\SHull(C))_{Res}=C_{Res} \cap (C_{Tor})^{\perp_S} ~~\text{and}~~ (\SHull(C))_{Tor}=(C_{Res})^{\perp_S} \cap C_{Tor}.$$
   We know that $C_{Res} \subseteq C_{Tor}$ if $C$ is $E$-linear. Then, $(C_{Tor})^{\perp_S} \subseteq (C_{Res})^{\perp_S}$. Consequently, $C_{Res} \cap (C_{Tor})^{\perp_S} \subseteq C_{Res} \cap (C_{Res})^{\perp_S}$. Thus, $(\SHull(C))_{Res} \subseteq \SHull(C_{Res})$. Further, $C_{Res} \subseteq C_{Tor}$ implies that $(C_{Res})^{\perp_S} \cap C_{Res} \subseteq (C_{Res})^{\perp_S} \cap C_{Tor}$. Thus, $(\SHull(C))_{Tor} \supseteq \SHull(C_{Tor})$. Next, if $C$ is free, then the freeness of the code $C$ implies that
$$(\SHull(C))_{Res}=C_{Res} \cap (C_{Tor})^{\perp_S}=C_{Res} \cap (C_{Res})^{\perp_S}=\SHull(C_{Res}),$$
and 
     $$(\SHull(C))_{Tor}=(C_{Res})^{\perp_S} \cap C_{Tor}=(C_{Tor})^{\perp_S} \cap C_{Tor}=\SHull(C_{Tor}).$$
   
\end{proof}

For an $\mathbb{F}_q$-linear code $C$, $(C^{\perp_S})^\perp=C$. Similarly, we study the following results for the three symplectic duals of an $E$-linear code.

\begin{theorem}\label{Thm4e}
   An $E$-linear code $C$ of length $2n$ satisfies the following duality properties:
    \begin{enumerate}
        \item If $C$ is free,  $(C^{\perp_{S_L}})^{\perp_{S_L}}=C$; 
        \item If $C_{Res}=\{0\}$ and $C_{Tor}=\mathbb{F}_2^{2n}$, $(C^{\perp_{S_R}})^{\perp_{S_R}}=C$; 
        \item $(C^{\perp_S})^{\perp_S}=C$.
       
        \end{enumerate}
\end{theorem}
\begin{proof}
    \begin{enumerate}
        \item For an $E$-linear code $C$, Corollary \ref{Cor3}(1) implies that
\begin{align*}
    (C^{\perp_{S_L}})^{\perp_{S_L}} &=\kappa ((C^{\perp_{S_L}})_{Res})^{\perp_S} \oplus \zeta ((C^{\perp_{S_L}})_{Res})^{\perp_S} \\
    &=\kappa ((C_{Res})^{\perp_S})^{\perp_S} \oplus \zeta ((C_{Res})^{\perp_S})^{\perp_S} \\
    &=\kappa C_{Res} \oplus \zeta C_{Res}.
    \end{align*}
This shows that $((C^{\perp_{S_L}})^{\perp_{S_L}})_{Res}=C_{Res}=((C^{\perp_{S_L}})^{\perp_{S_L}})_{Tor}$. Then the freeness of $C$ implies that $(C^{\perp_{S_L}})^{\perp_{S_L}}=C$.

        \item By Corollary \ref{Cor3}(2), we have 

         $$(C^{\perp_{S_R}})^{\perp_{S_R}} =\kappa ((C^{\perp_{S_R}})_{Tor})^{\perp_S} \oplus \zeta \mathbb{F}_2^{2n}=\kappa \{0\}^{\perp_S} \oplus \zeta \mathbb{F}_2^{2n}=\zeta \mathbb{F}_2^{2n}.$$
         This implies that $((C^{\perp_{S_R}})^{\perp_{S_R}})_{Res}=\{0\}$ and $((C^{\perp_{S_R}})^{\perp_{S_R}})_{Tor}=\mathbb{F}_2^{2n}$. 

        \item For an $E$-linear code $C$, Corollary \ref{Cor3}(3), we have
\begin{align*}
    (C^{\perp_S})^{\perp_{\perp_S}} &=\kappa ((C^{\perp_S})_{Tor})^{\perp_S} \oplus \zeta ((C^{\perp_S})_{Res})^{\perp_S} \\
    &=\kappa ((C_{Res})^{\perp_S})^{\perp_S} \oplus \zeta ((C_{Tor})^{\perp_S})^{\perp_S} \\
    &=\kappa C_{Res} \oplus \zeta C_{Tor} \\
    &=C.
    \end{align*}

    \end{enumerate}
\end{proof}

Here, we investigate the equality of left and two-sided symplectic duals of an $E$-linear code. 
\begin{theorem}\label{Thm3b}
    For an $E$-linear code $C$ of length $2n$, its left and two-sided symplectic duals are identical if and only if $C$ is left symplectic nice. Equivalently, these two symplectic duals are equal if and only if $C$ is free.
\end{theorem}
\begin{proof}
    If $C$ is an $E$-linear code of length $2n$, then $C^{\perp_{S_L}}=\kappa (C_{Res})^{\perp_S} \oplus \zeta (C_{Res})^{\perp_S}$ implies that $$|C^{\perp_{S_L}}|=|(C_{Res})^{\perp_S}| \cdot|(C_{Res})^{\perp_S}|=(2^{2n-k_1})^2=2^{4n-2k_1}.$$
    Hence, \begin{align}
        |C|\cdot|C^{\perp_{S_L}}|=2^{2k_1+k_2} \cdot 2^{4n-2k_1}=4^{2n}2^{k_2}.
        \end{align}
    This shows that if $C$ is left symplectic nice, then $k_2=0$. Hence, $C$ is free. Thus, 
    $$C^{\perp_{S_L}}=\kappa (C_{Res})^{\perp_S} \oplus \zeta (C_{Res})^{\perp_S}=\kappa (C_{Tor})^{\perp_S} \oplus \zeta (C_{Res})^{\perp_S}=C^{\perp_S}.$$
    Conversely, if $C^{\perp_{S_L}}=C^{\perp_S}$, then $(C_{Res})^{\perp_S}=(C_{Tor})^{\perp_S}$ implies that $C_{Res}=C_{Tor}$. Hence, $C$ is free and $k_2=0$. Thus, $C$ is left symplectic nice by Equation (1). \par
    Moreover, from the above discussions, it can be concluded that $C$ is left symplectic nice if and only if it is free. Then the second argument is immediate.
\end{proof}

 We now have the following analogous condition for the right symplectic dual.
\begin{theorem}\label{Thm3c}
   The right and two-sided symplectic duals of an $E$-linear code $C$ of length $2n$ are equal if and only if $C$ is right symplectic nice.
\end{theorem}
\begin{proof}
    We know that $C^{\perp_{S_R}}=\kappa (C_{Tor})^{\perp_S} \oplus \zeta \mathbb{F}_2^{2n}$. Then
   
        $$|C^{\perp_{S_R}}|  = |(C_{Tor})^{\perp_S}| \cdot|\mathbb{F}_2^{2n}| =2^{2n-\dim(C_{Tor})} \cdot 2^{2n} =2^{2n-(k_1+k_2)} \cdot 2^{2n} =2^{4n-(k_1+k_2)}.$$
    
    Therefore, $$|C|\cdot|C^{\perp_{S_R}}|=2^{2k_1+k_2} \cdot 2^{4n-(k_1+k_2)}=4^{2n}2^{k_1}.$$
    This shows that $C$ is right symplectic nice if and only if $k_1=0$, i.e., $C_{Res}=\{0\}.$ Further, if $C$ is right symplectic nice, then 
    
       $$ C^{\perp_S}=\kappa (C_{Tor})^{\perp_S} \oplus \zeta (C_{Res})^{\perp_S}=\kappa (C_{Tor})^{\perp_S} \oplus \zeta \mathbb{F}_2^{2n}=C^{\perp_{S_R}}.$$
       Conversely, if $C^{\perp_{S_R}}=C^{\perp_S}$, then $(C_{Res})^{\perp_S}=\mathbb{F}_2^{2n}$. This implies that $C_{Res}=\{0\}$, and hence $C$ is right symplectic nice.  
   
\end{proof}

\begin{theorem}
    A non-zero $E$-linear code can never be both left and right symplectic nice. 
\end{theorem}
\begin{proof}
    Given an $E$-linear code $C$ of length $2n$, Theorems \ref{Thm3b} and \ref{Thm3c} imply that
     $$|C|\cdot|C^{\perp_{S_L}}|=4^{2n}2^{k_2}~~~\text{and}~~~|C|\cdot|C^{\perp_{S_R}}|=4^{2n}2^{k_1}.$$
     Consequently, if $C$ is both left and right symplectic nice, then $k_1=k_2=0$. Therefore, $C_{Res}=C_{Tor}=\{0\}$, and hence $C=\{0\}$. Thus, any non-zero $E$-linear $C$ cannot be both left and right symplectic nice simultaneously.
    
\end{proof}

\begin{theorem}
    Every $E$-linear code is two-sided symplectic nice.
\end{theorem}
\begin{proof}
    Given an $E$-linear code $C$ of length $2n$, we have 
  $$|C||C^{\perp_S}|=|C_{Res}||C_{Tor}||(C_{Tor})^{\perp_S}||(C_{Res})^{\perp_S}|=4^{2n}.$$
    Thus, $C$ is always symplectic nice.
\end{proof}

We know that $\SHull(C^{\perp_S})=\SHull(C)$ if $C$ is an $\mathbb{F}_q$-linear code. Correspondingly, we investigate this for the three symplectic hulls of an $E$-linear code.

\begin{theorem}
    An $E$-linear code $C$ of length $2n$ satisfies the following properties:
    \begin{enumerate}
        \item If $C$ is free, then $\LSHull(C^{\perp_{S_L}})=\LSHull(C)$;
        \item If $C_{Res} =\{0\}$ and $C_{Tor}=\mathbb{F}_2^{2n}$, then $\RSHull(C^{\perp_{S_R}})=\RSHull(C)$;
        \item  $\SHull(C^{\perp_S})=\SHull(C)$.
 
\end{enumerate}
\end{theorem}

\begin{proof}
\begin{enumerate}
    \item For a free $E$-linear code $C$,  Theorem \ref{Thm4e}(1) implies that  $(C^{\perp_{S_L}})^{\perp_{S_L}}=C$. Consequently, $$\LSHull(C^{\perp_{S_L}})=C^{\perp_{S_L}} \cap (C^{\perp_{S_L}})^{\perp_{S_L}}=C^{\perp_{S_L}} \cap C=\LSHull(C).$$ 

    \item Given an $E$-linear code $C$ of length $2n$ with $C_{Res}=\{0\}$ and $C_{Tor}=\mathbb{F}_2^{2n}$, Theorem \ref{Thm4e}(2) implies that $(C^{\perp_{S_R}})^{\perp_{S_R}}=C$. Therefore,
    $$\RSHull(C^{\perp_{S_R}})=C^{\perp_{S_R}} \cap (C^{\perp_{S_R}})^{\perp_{S_R}}=C \cap C^{\perp_{S_R}}=\RSHull(C).$$

    \item For an $E$-linear code $C$,  $(C^{\perp_S})^{\perp_S}=C$ by Theorem \ref{Thm4e}(3). This shows that $$\SHull(C^{\perp_S})= C^{\perp_S}\cap (C^{\perp_S})^{\perp_S}= C \cap C^{\perp_S}=\SHull(C).$$
\end{enumerate}    
\end{proof}

Next, we have the following result, which investigates the freeness of the hull code of a free $E$-linear code.

\begin{theorem}\label{Thm3a}
 For a free $E$-linear code $C$, its symplectic hull code is also free.   
\end{theorem}
\begin{proof}
    If $C$ is a free $E$-linear code, then by (3) of Theorem \ref{Thm2}, we have 
  $$(\SHull(C))_{Res}=C_{Res} \cap (C_{Tor})^{\perp_S},~~~\text{and}~~~(\SHull(C))_{Tor}=(C_{Res})^{\perp_S} \cap C_{Tor}.$$
Since $C$ is free, $C_{Res}=C_{Tor}$. Therefore,
    $(\SHull(C))_{Res}=(\SHull(C))_{Tor}.$
 Thus, the symplectic hull of the free $E$-code $C$ is also free.
 \end{proof}

For a free $E$-linear code, we have the following example, which shows that its right symplectic hull code may not be free.

\begin{example}
  If $C$ is an $E$-linear code with generator matrix 
  $$G=\begin{pmatrix}
     \kappa & 0 & 0 & 0\\
      0  & 0  & \kappa & 0
  \end{pmatrix},$$
  then $C$ is free, as $C_{Res}$ and $C_{Tor}$ both are generated by the matrix 
  $$G_1=\begin{pmatrix}
      1 & 0 & 0& 0\\
      0 & 0 & 1 & 0
  \end{pmatrix}.$$
 Next, the symplectic dual $(C_{Res})^{\perp_S}$ is generated by the matrix
  $$H_1=\begin{pmatrix}
      0 & 1 & 0& 0\\
      0 & 0 & 0 & 1
  \end{pmatrix}.$$
  Clearly, $C_{Res} \cap (C_{Res})^{\perp_S}=\{0\}$. Therefore, from Theorem \ref{Thm2}(2), we have $$(\RSHull(C))_{Res}=C_{Res} \cap (C_{Tor})^{\perp_S}=C_{Res} \cap (C_{Res})^{\perp_S}=\{0\},$$
  and $$ (\RSHull(C))_{Tor}=C_{Tor}.$$
  Therefore, $(\RSHull(C))_{Res} \neq (\RSHull(C))_{Tor}$. Thus, $\RSHull(C)$ is not free.
\end{example}

 The following result is crucial in our further investigations on symplectic hulls.
\begin{theorem}
   For a free $E$-linear code $C$, $\LSHull(C)=\SHull(C)$.
\end{theorem}
\begin{proof}
 By Theorem \ref{Thm3b}, if $C$ is a free $E$-linear code, then its left and two-sided symplectic duals are equal. Therefore, its left and two-sided symplectic hulls are also equal.
\end{proof}

By the preceding results, the left symplectic hull and the two-sided symplectic hull of a free $E$-linear code coincide. Additionally, they are free. However, there is no guarantee that the right symplectic hull is free. Consequently, we restrict our attention to the two-sided symplectic hull of a free $E$-linear code.\par

We now have the following result that characterizes a generator matrix of the symplectic hull of a free $E$-linear code.

\begin{theorem}\label{thm1a}
  Let $G$ be a generator matrix of the symplectic hull of the residue code of a free $E$-linear code $C$. Then the symplectic hull $\SHull(C)$ has a generator matrix $\kappa G$.
 
\end{theorem}
\begin{proof}
  If $C$ is a free $E$-code, then by Theorem  \ref{Thm3a}, its symplectic hull $\SHull(C)$ is also free. In addition, $(\SHull(C))_{Res}=\SHull(C_{Res})$ from Corollary \ref{Cor1}. Further, Theorem $3.5$ of \cite{Kushwaha2} implies that if $G'$ is a generator matrix of the residue code of a free $E$-linear code $C$, then $\kappa G'$ is a generator matrix of $C$. Therefore, if $G$ is a generator matrix of the symplectic hull of $C_{Res}$, then $\kappa G$ is a generator matrix of $\SHull(C)$. 
\end{proof}

The next result is essential for classifying optimal free $E$-linear codes, as it calculates the symplectic hull-rank.
\begin{theorem}\label{thm3}
   For a free $E$-linear code $C$, 
   $$\rank(\SHull(C))=\dim(\SHull(C_{Res})).$$
\end{theorem}
\begin{proof}
   The proof follows immediately from the definition of the rank of a free $E$-linear code together with Theorem \ref{thm1a}.
\end{proof}

  The sum of two $E$-linear codes $C$ and $D$ is a code defined by  $$C+D=\{x+y~|~ x\in C, y \in D\}.$$
  Now, we investigate the three symplectic duals of the sum of two $E$-linear codes.
\begin{theorem}\label{Thm15}
    For two $E$-linear codes $C$ and $D$, 
    \begin{enumerate}
        \item $(C+D)^{\perp_S}=C^{\perp_S} \cap D^{\perp_S}$;
       \item  $(C+D)^{\perp_{S_L}}=C^{\perp_{S_L}} \cap D^{\perp_{S_L}}$;
       \item $(C+D)^{\perp_{S_R}}=C^{\perp_{S_R}} \cap D^{\perp_{S_R}}$.
        
    \end{enumerate}
    \end{theorem}
    \begin{proof}
        \begin{enumerate}
            \item Let $x \in (C+D)^{\perp_S}$. Then $\langle x,y \rangle_s=0$ for all $y \in C+D$. Since $C, D \subseteq C+D$, $\langle x,y \rangle_s=0$ for all $y \in C$ and $\langle x,y \rangle_s=0$ for all $y \in D$. Hence, $x \in C^{\perp_S}$ and $x \in D^{\perp_S}$. Consequently, $x \in C^{\perp_S} \cap D^{\perp_S}$. Therefore, $(C+D)^{\perp_S} \subseteq C^{\perp_S} \cap D^{\perp_S}$. Next, suppose that $x \in C^{\perp_S} \cap D^{\perp_S}$. Since $ x \in C^{\perp_S}$ and $ x \in D^{\perp_S}$, we have 
           $$ \langle x,y \rangle_s=0~~ \text{for all}~~ y \in C~~ \text{and}~~ \langle x,z \rangle_s=0~~ \text{for all}~~ z \in D.$$
            Hence, $\langle x,y \rangle_s+\langle x,z \rangle_s=0$ for all $y \in C$ and $z \in D$. This implies that $\langle x,y+z\rangle_s=0$ for all $y \in C$ and $z \in D$. Therefore, $x \in (C+D)^{\perp_S}$, and so $C^{\perp_S} \cap D^{\perp_S} \subseteq (C+D)^{\perp_S}$. Thus, $(C+D)^{\perp_S}=C^{\perp_S} \cap D^{\perp_S}$.\par
            
             The proofs of the remaining two parts are analogous to that of the first part.
        \end{enumerate}
    \end{proof}

 We now explore the symplectic hull of the sum of two free $E$-linear codes and compute the corresponding symplectic hull-rank.
 \begin{theorem}\label{Thm16}
     Let $C$ and $D$ be two free $E$-linear $[2n,k]$ and $[2n,k']$-codes with symplectic hull-ranks $l_1$ and $l_2$, respectively. Also, let $C \subseteq D^{\perp_S}$, and $\rank(\SHull(C)\cap \SHull(D))=l$. Then
     $$\SHull(C+D)=\SHull(C)+\SHull(D),$$
     and $$ \rank (\SHull(C+D))=l_1+l_2-l.$$
 \end{theorem}
 \begin{proof}
    Let $C$ and $D$ be two free $E$-linear codes. Assume that $x\in\SHull(C+D)$. Then
$x\in C+D$, and hence there exist $u\in C$ and
$v\in D$ such that $x=u+v$. Further, $x \in (C+D)^{\perp_S}$ and from Theorem \ref{Thm15}, $(C+D)^{\perp_S}=C^{\perp_S} \cap D^{\perp_S}$. Consequently,
$x\in C^{\perp_S}\cap D^{\perp_S}$. Since $D\subseteq C^{\perp_S}$ and
$C\subseteq D^{\perp_S}$, we have
$v\in C^{\perp_S}$ and $u\in D^{\perp_S}$.
Now, for all $a\in C$, it follows that
$$
0=\langle x,a\rangle
=\langle u,a\rangle+\langle v,a\rangle
=\langle u,a\rangle.
$$
Then, $u\in C^{\perp_S}$, and hence
$u\in \SHull(C)$. Similarly, for all $b\in D$,
$$
0=\langle x,b\rangle
=\langle u,b\rangle+\langle v,b\rangle
=\langle v,b\rangle.
$$
Hence, $v\in D^{\perp_S}$, and so $v\in \SHull(D)$. Consequently, $x \in \SHull(C)+\SHull(D)$. Therefore,
$$
\SHull(C+D)\subseteq \SHull(C)+\SHull(D).
$$

Conversely, let $x\in \SHull(C)+\SHull(D)$. Then there exist
$u\in \SHull(C)$ and $v\in \SHull(D)$ such that $x=u+v$.
Since $u\in C$ and $v\in D$, we have
$x\in C+D$. Moreover, $u\in C^{\perp_S}$
and $v\in D^{\perp_S}$. Then $C \subseteq D^{\perp_S}$ and $D \subseteq C^{\perp_S}$ imply that $u,v\in C^{\perp_S}\cap D^{\perp_S}$. Then, $x\in C^{\perp_S}\cap D^{\perp_S}$. From Theorem \ref{Thm15}, $(C+D)^{\perp_S}=C^{\perp_S} \cap D^{\perp_S}$. Consequently, $x\in (C+D)^{\perp_S}$. Therefore, $x\in\SHull(C+D)$,
which proves that $\SHull(C)+\SHull(D)\subseteq\SHull(C+D)$. Thus,
$$\SHull(C+D)=\SHull(C)+\SHull(D),$$
    and \begin{align*}    \rank(\SHull(C+D))&=\rank(\SHull(C))+\rank(\SHull(D))\\&\hspace{1cm} -\rank(\SHull(C) \cap \SHull(D))\\
         &=l_1+l_2-l.
    \end{align*}
 \end{proof}

\begin{corollary}
   Let $C$ and $D$ be two symplectic LCD $E$-codes such that $C\subseteq D^{\perp_S}$. Then $C+D$ is also a symplectic LCD $E$-code.
\end{corollary}

\begin{proof}
Since $C$ and $D$ are symplectic LCD codes, we have
$$
\SHull(C)=\{0\}
\quad\text{and}\quad
\SHull(D)=\{0\}.
$$
Then, by Theorem \ref{Thm16}, we obtain
$$
\SHull(C+D)
=\SHull(C)+\SHull(D)
=\{0\}.
$$
Therefore, $C+D$ is also an LCD code.
\end{proof}

\section{Build-up constructions}
This section presents two build-up construction techniques that construct free $E$-linear codes with a larger length and symplectic hull-rank from free $E$-linear codes with a smaller length and symplectic hull-rank. To support our build-up construction methods, we also provide illustrative examples of codes constructed using these methods. \par

The following result from \cite{Li2024} is essential for our build-up construction methods.

\begin{theorem}[ \cite{Li2024}, Theorem $3.1$]\label{Thm4} Let $G$ be a generator matrix of a $k$-dimensional $\mathbb{F}_q$-linear code $C$  of length $2n$ with symplectic hull-dimension $h$. Then, $h=k-\rank(G\Omega_{2n} G^T)$ where $G^T$ denotes the transpose of $G$. 
\end{theorem}

Here, we present our first build-up construction method.\par

\begin{theorem}[\textbf{Construction I}] Let $C$ be a free $E$-linear $[2n,k]$-code with sympathetic hull-rank $l$, and $G$ (with rows $r_i=(c_{i,1},c_{i,2}, \ldots c_{i,n} |d_{i,1},d_{i,2}, \ldots d_{i,n})$ for $1 \leq i \leq k$) be a generator matrix of its residue code $C_{Res}$. Also, assume that $x=(a|b)=(a_1, a_2, \ldots, a_n | b_1, b_2, \ldots, b_n) \in \mathbb{F}_2^{2n}$ such that $\langle x, r_i \rangle_s =0$ for $1 \leq i \leq k$. 
Then 
\begin{enumerate}
    \item[(a)] The matrix $G'$ given below generates a free $E$-linear $[2n+2, k+1]$-code $D$ with symplectic hull-rank $l+1$:

   $$
G' =
\left(
\setlength{\arraycolsep}{3pt}
\begin{array}{cccccccccc}
\kappa & \kappa a_1 & \kappa a_2 &   \cdots & \kappa a_n & \kappa & \kappa b_1 & \kappa b_2 &   \cdots & \kappa b_n \\ 
0 &\kappa c_{1,1} & \kappa c_{1,2} & \cdots &  \kappa c_{1,n} & 0 & \kappa d_{1,1} & \kappa d_{1,2} & \cdots &  \kappa d_{1,n}\\
0 &\kappa c_{2,1} & \kappa c_{2,2} & \cdots &  \kappa c_{2,n} & 0 & \kappa d_{2,1} & \kappa d_{2,2} & \cdots &  \kappa d_{2,n}\\
\vdots & \vdots & \vdots & &  \vdots & \vdots & \vdots & \vdots & & \vdots \\
0 &\kappa c_{k,1} & \kappa c_{k,2} & \cdots &  \kappa c_{k,n} & 0 & \kappa d_{k,1} & \kappa d_{k,2} & \cdots &  \kappa d_{k,n}\\
\end{array}
\right)
.$$

    \item[(b)] Let $H$ (with rows $s_j=(s_{j,1},s_{j,2}, \ldots s_{j,n} |t_{j,1},t_{j,2}, \ldots t_{j,n})$ for $1 \leq j \leq 2n-k$) be a parity-check matrix of its residue code $C_{Res}$. Also, let $y=(u|v)=(u_1,u_2, \ldots, u_n | v_1, v_2, \ldots, v_n) \in (C_{Res})^{\perp_S}$, and $\delta=\langle x,y \rangle_s$. Further, assume that $z_j=\langle x, s_j \rangle_s$ for $1 \leq j \leq m$ where $m=2n-k$. Then the code $D$ has the following parity-check matrix:

$$
H' =
\left(
\setlength{\arraycolsep}{2.5pt}
\begin{array}{cccccccccc}

\kappa & \kappa u_1 & \kappa u_2 &   \cdots & \kappa u_n & \kappa(1+\delta) & \kappa v_1 & \kappa v_2 &   \cdots & \kappa v_n \\ 
0 &\kappa s_{1,1} & \kappa s_{1,2} & \cdots &  \kappa s_{1,n} & z_1 & \kappa t_{1,1} & \kappa t_{1,2} & \cdots &  \kappa t_{1,n}\\
0 &\kappa s_{2,1} & \kappa s_{2,2} & \cdots &  \kappa s_{2,n} & z_2 & \kappa t_{2,1} & \kappa t_{2,2} & \cdots &  \kappa t_{2,n}\\
\vdots & \vdots & \vdots & &  \vdots & \vdots & \vdots & \vdots & & \vdots \\
0 &\kappa s_{m,1} & \kappa s_{m,2} & \cdots &  \kappa s_{m,n} & z_m & \kappa t_{m,1} & \kappa t_{m,2} & \cdots &  \kappa t_{m,n}\\
\end{array}
\right)
.$$
    
\end{enumerate}
\end{theorem}
\begin{proof}

\begin{enumerate}
    \item[(a)]
   It follows that $G'$ generates a free $E$-linear code, and its residue code has the generator matrix
    $$
G_1 =
\left(
\begin{array}{cccccccccc}
1 &  a_1 &  a_2 &   \cdots &  a_n & 1 &  b_1 &  b_2 &   \cdots &  b_n \\ 
0 & c_{1,1} &  c_{1,2} & \cdots &   c_{1,n} & 0 & d_{1,1} &  d_{1,2} & \cdots &   d_{1,n}\\
0 & c_{2,1} &  c_{2,2} & \cdots &   c_{2,n} & 0 &  d_{2,1} &  d_{2,2} & \cdots &   d_{2,n}\\
\vdots & \vdots & \vdots & &  \vdots & \vdots & \vdots & \vdots & & \vdots \\
0 & c_{k,1} &  c_{k,2} & \cdots &   c_{k,n} & 0 &  d_{k,1} &  d_{k,2} & \cdots &  d_{k,n}\\
\end{array}
\right)
.$$
If $r_i'$ denotes the $i${th} row of $G_1$ for $1\leq i\leq k+1$, then

$$
G_1\Omega_{2n+2}G_1^T=
\left(
\setlength{\arraycolsep}{3pt}
\begin{array}{cccc}
\langle r_1',r_1' \rangle_s & \langle r_1',r_2' \rangle_s & \cdots \cdots & \langle r_1',r_{k+1}' \rangle_s \\
\langle r_2',r_1' \rangle_s & \langle r_2',r_2' \rangle_s & \cdots \cdots & \langle r_2',r_{k+1}' \rangle_s \\
\vdots & \vdots & & \vdots \\
\langle r_{k+1}',r_1' \rangle_s & \langle r_{k+1}',r_2' \rangle_s & \cdots \cdots & \langle r_{k+1}',r_{k+1}' \rangle_s 
\end{array}
\right)
.$$

The entries in the first row and in the first column of $G_1\Omega_{2n+2}G_1^T$ are the symplectic inner products of the first row of $G_1$ with its other rows. Then, for $2\leq i \leq k+1$, we have
$$\langle r_1', r_i' \rangle_s=1\cdot 0+ 0 \cdot 1 + \sum_{j=1}^n (a_jd_{i,j}+b_jc_{i,j})=\langle x, r_{i-1} \rangle_s=0.$$
Also, $\langle r_i', r_1' \rangle_s=0$, as symplectic inner product is symmetric in $\mathbb{F}_2^{2n}$. Further, the symplectic inner product of the top row of $G_1$ with itself is 
$$\langle r_1', r_1' \rangle_s=1 \cdot 1 +1 \cdot 1 + \sum_{j=1}^n(a_jb_j+b_ja_j)=0.$$
Since all the above discussed symplectic inner products are zero, the first row and the first column of the matrix  $G_1\Omega_{2n+2}G_1^T$ are also zero. On the other hand, the symplectic inner product of other rows $r_i'$ and $r_j'$ for $2  \leq i,j \leq k+1$, i.e., any two rows from the second row to the $(k+1)^{th}$ row of $G_1$, with each other is given by 

$$ \langle r_i', r_j' \rangle_s=0 \cdot 0+ 0 \cdot 0 + \sum_{h=1}^n(c_{i,h}d_{j,h}+d_{i,h}c_{j,h})= \langle r_{i-1}, r_{j-1} \rangle_s.$$

Hence, the matrix $G_1\Omega_{2n+2}G_1^T$ reduces to
$$
G_1\Omega_{2n+2}G_1^T=
\left(
\begin{array}{ccccccc}
0 & 0 &  0 &    \cdots \cdots & & 0 \\  
 0 &   & & &   & & \\
 0 &   & & G \Omega_{2n}G^T&  & & \\
\vdots &  & & & & & \\
0 &  & & &   & &
\end{array}
\right)
.$$
Therefore, $\rank(G_1 \Omega_{2n+2} G_1^T)=\rank(G \Omega_{2n} G^T)$. Since $C$ is free, by Theorem \ref{thm3}, we have $\dim(\SHull(C_{Res}))=\rank(\SHull(C))=l$.  Then, by Theorem \ref{Thm4}, $\rank(G \Omega_{2n} G^T)=k-l$ and so $\rank(G_1 \Omega_{2n+2} G_1^T)=k-l$. Again, Theorem \ref{Thm4} implies that the symplectic hull-dimension of $D_{Res}$ is given by $(k+1)-\rank(G_1\Omega_{2n+2} G_1^T)=(k+1)-(k-l)=l+1$. For a free $E$-linear code $D$, $\rank(D)=\dim(D_{Res})$. Therefore,  $\rank(\SHull(D))=\dim(\SHull(D_{Res}))=l+1$. 
\item[(b)] We have  
$$
\pi(H') =H_1=
\left(
\setlength{\arraycolsep}{2.7pt}
\begin{array}{cccccccccc}

1 &  u_1 &  u_2 &   \cdots & u_n & (1+\delta) &  v_1 &  v_2 &   \cdots &  v_n \\ 
0 & s_{1,1} &  s_{1,2} & \cdots &   s_{1,n} & z_1 &  t_{1,1} &  t_{1,2} & \cdots &  t_{1,n}\\
0 & s_{2,1} &  s_{2,2} & \cdots &   s_{2,n} & z_2 &  t_{2,1} &  t_{2,2} & \cdots &   t_{2,n}\\
\vdots & \vdots & \vdots & &  \vdots & \vdots & \vdots & \vdots & & \vdots \\
0 & s_{m,1} &  s_{m,2} & \cdots &   s_{m,n} & z_m &t_{m,1} &  t_{m,2} & \cdots &   t_{m,n}\\
\end{array}
\right)
.$$

  Our aim is to show $G_1 \Omega_{2n+2}H_1^T=0$. In order to prove this, we calculate the following symplectic inner products:
\begin{itemize}
    \item Top row of $G_1$ with top row of $H_1$:
    \begin{align*}
    \langle (1,a|1,b), (1,u |1+\delta,v) \rangle_s &=1(1+\delta)+1\cdot 1 + \langle a, v \rangle + \langle b,u \rangle \\
    &= \delta + \langle (a|b), (u|v) \rangle_s\\
    &=\delta+\delta \\
    &=0.
    \end{align*}

    \item Top row of $G_1$ with other rows of $H_1$:
    \begin{align*}
        \langle (1,a| 1,b), (0,s_{i,j} |z_i,t_{i,j}) \rangle_s & =1 \cdot z_i+1 \cdot 0 + \langle a, t_{i,j} \rangle+ \langle b, s_{i,j} \rangle \\
        &=z_i+\langle (a|b), (s_{i,j},t_{i,j}) \rangle_s\\
        & = z_i+z_i \\
        & =0.
    \end{align*}
    \item Other rows of $G_1$ with top row of $H_1$:
    \begin{align*}
        \langle (0,c_{i,j}|0,d_{i,j}), (1,u| 1+\delta,v) \rangle_s & = 0 \cdot(1+\delta)+0 \cdot 1+ \langle c_{i,j},v \rangle  \\
        &\hspace{1cm}+\langle d_{i,j}, u \rangle \\
        &= \langle (c_{i,j}|d_{i,j}), (u|v) \rangle_s \\
        &=0.~~~~~~(\because~~ (u|v) \in (C_{Res})^{\perp_S})
    \end{align*}

    \item Other rows of $G_1$ with other rows of $H_1$:
    \begin{align*}
        \langle (0,c_{i,j}|0,d_{i,j}), (0,s_{i,j} |z_i,t_{i,j}) \rangle_s & =0 \cdot z_i+ \langle c_{i,j}, t_{i,j} \rangle + \langle d_{i,j}, s_{i,j} \rangle \\
        &=\langle (c_{i,j}|d_{i,j}), (s_{i,j}|t_{i,j}) \\
        &=0.
    \end{align*}
 
\end{itemize}
 From all the above cases, we conclude that $G_1 \Omega_{2n+2}H_1^T=0$. Further, the dimension of the code generated by $H_1$ is $(2n-k)+1$, since the first row of $H_1$ is not a linear combination of the other rows of $H_1$. Also, 
$$\dim \langle H_1 \rangle =(2n-k)+1=(2n+2)-(k+1)=\dim((D^\perp)_{Res}).$$
Therefore, $H_1$ is a parity-check matrix of $(D)_{Res}$. Thus, $H'$ is a parity-check matrix of $D$.

\end{enumerate}
\end{proof}

Now, we have our second build-up construction method as follows. 

\begin{theorem}[\textbf{Construction II}]Let $C$ be a free $E$-linear $[2n,k]$-code with symplectic hull-rank $l$. Also, let $G$  be a generator matrix of its residue code $C_{Res}$ with rows $r_i=(c_{i,1},c_{i,2}, \ldots c_{i,n} |d_{i,1},d_{i,2}, \ldots d_{i,n})$ for $1 \leq i \leq k$. Furthermore, fix a vector $x=(a|b)=(a_1, a_2, \ldots, a_n | b_1, b_2, \ldots, b_n) \in \mathbb{F}_2^{2n}$ with $\langle x, r_i \rangle_s =0$ for $1 \leq i \leq k$. 
Then 
\begin{enumerate}
    \item[(a)] The free $E$-linear code $D$ with generator matrix $G'$ given below is a $[2n+4, k+1]$-code with symplectic hull-rank $l+1$:

   $$
G' =
\left(
\setlength{\arraycolsep}{3pt}
\begin{array}{cccccccccccc}
\kappa & \kappa  & \kappa a_1 & \kappa a_2 &   \cdots & \kappa a_n & \kappa & \kappa & \kappa b_1 & \kappa b_2 &   \cdots & \kappa b_n \\ 
0&0 &\kappa c_{1,1} & \kappa c_{1,2} & \cdots &  \kappa c_{1,n} & 0&0 & \kappa d_{1,1} & \kappa d_{1,2} & \cdots &  \kappa d_{1,n}\\
0&0 &\kappa c_{2,1} & \kappa c_{2,2} & \cdots &  \kappa c_{2,n} & 0&0 & \kappa d_{2,1} & \kappa d_{2,2} & \cdots &  \kappa d_{2,n}\\
\vdots & \vdots & \vdots &   \vdots & & \vdots & \vdots & \vdots & \vdots & \vdots & & \vdots \\
0&0 &\kappa c_{k,1} & \kappa c_{k,2} & \cdots &  \kappa c_{k,n} & 0&0 & \kappa d_{k,1} & \kappa d_{k,2} & \cdots &  \kappa d_{k,n}\\
\end{array}
\right)
.$$
    \item[(b)] Let $H$ (with rows $s_j=(s_{j,1},s_{j,2}, \ldots s_{j,n} |t_{j,1},t_{j,2}, \ldots t_{j,n}) ~~ \text{for}~~ 1 \leq j \leq 2n-k$) be a parity-check matrix of its residue code $C_{Res}$. Also, let $y=(u|v)=(u_1,u_2, \ldots, u_n | v_1, v_2, \ldots, v_n) \in (C_{Res})^{\perp_S}$ be a non-zero vector with $\langle x,y \rangle_s=0$. Further, assume that $z_j=\langle x, s_j \rangle_s=0$ for $1 \leq j \leq m$, where $m=2n-k$. Then the code $D$ has a parity-check matrix as follows:

$$
H' =
\left(
\setlength{\arraycolsep}{3pt}
\begin{array}{cccccccccccc}
\kappa & 0 & 0 & 0 & \cdots & 0 & 0 & \kappa & 0 & 0& \cdots & 0 \\
0 & \kappa & 0 & 0 & \cdots & 0 & \kappa & 0 & 0 & 0& \cdots & 0 \\
\kappa &0 & \kappa u_1 & \kappa u_2 &   \cdots & \kappa u_n & \kappa &0 & \kappa v_1 & \kappa v_2 &   \cdots & \kappa v_n \\ 
0 &0 &\kappa s_{1,1} & \kappa s_{1,2} & \cdots &  \kappa s_{1,n} & 0 &0 & \kappa t_{1,1} & \kappa t_{1,2} & \cdots &  \kappa t_{1,n}\\
0 &0 &\kappa s_{2,1} & \kappa s_{2,2} & \cdots &  \kappa s_{2,n} & 0 &0 & \kappa t_{2,1} & \kappa t_{2,2} & \cdots &  \kappa t_{2,n}\\
\vdots & \vdots & \vdots & \vdots&   & \vdots & \vdots & \vdots & \vdots& \vdots & & \vdots \\
0 &0 &\kappa s_{m,1} & \kappa s_{m,2} & \cdots &  \kappa s_{m,n} & 0 &0 & \kappa t_{m,1} & \kappa t_{m,2} & \cdots &  \kappa t_{m,n}\\
\end{array}
\right)
.$$
    
\end{enumerate}
\end{theorem}

\begin{proof}

\begin{enumerate}
    \item[(a)] We see that the matrix $G'$ generates a free $E$-linear code, and its residue code has the generator matrix
   $$
G_1 =
\left(
\begin{array}{cccccccccccc}
1 & 1  &  a_1 &  a_2 &   \cdots &  a_n & 1 & 1 &  b_1 &  b_2 &   \cdots &  b_n \\ 
0&0 & c_{1,1} &  c_{1,2} & \cdots &   c_{1,n} & 0&0 &  d_{1,1} &  d_{1,2} & \cdots &   d_{1,n}\\
0&0 & c_{2,1} &  c_{2,2} & \cdots &   c_{2,n} & 0&0 &  d_{2,1} &  d_{2,2} & \cdots &   d_{2,n}\\
\vdots & \vdots & \vdots &   \vdots & & \vdots & \vdots & \vdots & \vdots & \vdots & & \vdots \\
0&0 & c_{k,1} &  c_{k,2} & \cdots &   c_{k,n} & 0&0 &  d_{k,1} &  d_{k,2} & \cdots &   d_{k,n}\\
\end{array}
\right)
.$$

Similar to the \textbf{Construction I}, if $r_i'$ is the $i${th} row of $G_1$ for $1\leq i\leq k+1$, then

$$
G_1\Omega_{2n+4}G_1^T=
\left(
\begin{array}{cccc}
\langle r_1',r_1' \rangle_s & \langle r_1',r_2' \rangle_s & \cdots \cdots & \langle r_1',r_{k+1}' \rangle_s \\
\langle r_2',r_1' \rangle_s & \langle r_2',r_2' \rangle_s & \cdots \cdots & \langle r_2',r_{k+1}' \rangle_s \\
\vdots & \vdots & & \vdots \\
\langle r_{k+1}',r_1' \rangle_s & \langle r_{k+1}',r_2' \rangle_s & \cdots \cdots & \langle r_{k+1}',r_{k+1}' \rangle_s 
\end{array}
\right)
.$$

The symplectic inner product of the top row of $G_1$ with itself is given by
$$\langle r_1', r_1' \rangle_s=1 \cdot 1 +1 \cdot 1 + 1 \cdot 1 +1 \cdot 1+ \sum_{j=1}^n(a_jb_j+b_ja_j)=0.$$
Next, the symplectic inner product of the top row of $G_1$ with all other rows is 
$$\langle r_1', r_{i+1}' \rangle_s=1\cdot 0+ 0 \cdot 1 +1\cdot 0+ 0 \cdot 1 + \sum_{j=1}^n (a_jd_{i,j}+b_jc_{i,j})=\langle x, r_i \rangle_s=0.$$
Also, $\langle r_{i+1}', r_1' \rangle_s=0$, as symplectic inner product is symmetric in $\mathbb{F}_2^{2n}$. Hence, the first row and the first column of the matrix  $G_1\Omega_{2n+4}G_1^T$ will be zero. Further, the symplectic inner product of other rows $r_i'$ and $r_j'$ for $2  \leq i,j \leq k+1$, i.e., any two rows from second row to the $(k+1)$th row of $G_1$ with each other is given by 

$$ \langle r_{i}', r_{j}' \rangle_s=0 \cdot 0+ 0 \cdot 0 + \sum_{h=1}^n(c_{i,h}d_{j,h}+d_{i,h}c_{j,h})= \langle r_{i-1}, r_{j-1} \rangle_s.$$

Therefore, the matrix $G_1\Omega_{2n+4}G_1^T$ reduces to
$$
G_1\Omega_{2n+4}G_1^T=
\left(
\begin{array}{ccccccc}
0 &  &   &    \cdots \cdots & & 0 \\  
 0 &   & & &   & & \\
 0 &   & & G \Omega_{2n}G^T&  & & \\
\vdots &  & & & & & \\
0 &  & & &   & &
\end{array}
\right)
.$$
This implies that $\rank(G_1 \Omega_{2n+4} G_1^T)=\rank(G \Omega_{2n} G^T)$. Since $C$ is free, by Theorem \ref{thm3}, we have $\dim(\SHull(C_{Res}))=\rank(\SHull(C))=l$.  Then, by Theorem \ref{Thm4}, $\rank(G \Omega_{2n} G^T)=k-l$ and so $\rank(G_1 \Omega_{2n+4} G_1^T)=k-l$. Again, Theorem \ref{Thm4} implies that the symplectic hull-dimension of $D_{Res}$ is given by $(k+1)-rank(G_1\Omega_{2n+4} G_1^T)=(k+1)-(k-l)=l+1$. Thus, we have $\rank(\SHull(D))=\dim(\SHull(D_{Res}))=l+1$. 

\item[(b)] We have  
\small$$
\pi(H') =H_1=
\left(
\setlength{\arraycolsep}{3pt}
\begin{array}{cccccccccccc}
1 & 0 & 0 & 0 & \cdots & 0 & 0 & 1 & 0 & 0& \cdots & 0 \\
0 & 1 & 0 & 0 & \cdots & 0 & 1 & 0 & 0 & 0& \cdots & 0 \\
1 &0 &  u_1 &  u_2 &   \cdots &  u_n & 1 &0 &  v_1 &  v_2 &   \cdots &  v_n \\ 
0 &0 & s_{1,1} &  s_{1,2} & \cdots &   s_{1,n} & 0 &0 &  t_{1,1} &  t_{1,2} & \cdots &   t_{1,n}\\
0 &0 & s_{2,1} &  s_{2,2} & \cdots &   s_{2,n} & 0 &0 &  t_{2,1} &  t_{2,2} & \cdots &   t_{2,n}\\
\vdots & \vdots & \vdots & \vdots&   & \vdots & \vdots & \vdots & \vdots & \vdots & & \vdots \\
0 &0 & s_{m,1} &  s_{m,2} & \cdots &   s_{m,n} & 0 &0 &  t_{m,1} &  t_{m,2} & \cdots &   t_{m,n}\\
\end{array}
\right)
.$$
We need to calculate the following symplectic inner products:
\begin{itemize}
    \item First row of $G_1$ with the first row of $H_1$:
    \begin{align*}
        \langle  (1,1,a|1,1,b), (1,0,\textbf{0}| 0,1,\textbf{0}) \rangle_s & =1 \cdot 0+1 \cdot 1+ \langle a,\textbf{0} \rangle \\
        &\hspace{0.4cm} + 1 \cdot 1+1 \cdot 0+ \langle b,\textbf{0} \rangle \\
        &=0.
    \end{align*}

     \item First row of $G_1$ with the second row of $H_1$:
    \begin{align*}
        \langle  (1,1,a|1,1,b), (0,1,\textbf{0}| 1,0,\textbf{0}) \rangle_s & =1 \cdot 1+1 \cdot 0+ \langle a,\textbf{0} \rangle \\
        &\hspace{0.4cm}+ 1 \cdot 0+1 \cdot 1+ \langle b,\textbf{0} \rangle \\
        &=0.
    \end{align*}

    \item First row of $G_1$ with the third row of $H_1$:
    \begin{align*}
        \langle  (1,1,a|1,1,b), (1,0, u| 1,0, v) \rangle_s & =1 \cdot 1+1 \cdot 0+ \langle a,v \rangle \\
        &\hspace{0.4cm}+ 1 \cdot 1+1 \cdot 0+ \langle b,u \rangle \\
        &= \langle (a|b), (u|v) \rangle_s\\
        &=0.
    \end{align*}

        \item First row of $G_1$ with the $i${th} row of $H_1$ where $i>3$:
    \begin{align*}
        \langle  (1,1,a|1,1,b), (0,0, s_{i,j}| 0,0, t_{i,j}) \rangle_s & =1 \cdot 0+1 \cdot 0+ \langle a,t_{i,j} \rangle \\
        &\hspace{0.4cm}+ 1 \cdot 0+1 \cdot 0+ \langle b,s_{i,j} \rangle \\
        &= \langle (a|b), (s_{i,j}|t_{i,j}) \rangle_s\\
        &= \langle (a|b), s_i\rangle_s\\
        &=0.
    \end{align*}

    \item For $i>1$, the $i${th} row of $G_1$ with the first row of $H_1$:
    \begin{align*}
        \langle   (0,0, c_{i,j}| 0,0, d_{i,j}), (1,0,\textbf{0}|0,1,\textbf{0}) \rangle_s & =0 \cdot 0+0 \cdot 1+ \langle c_{i,j}, \textbf{0} \rangle\\
        &\hspace{0.3cm}+ 0 \cdot 1+0 \cdot 0+ \langle d_{i,j}, \textbf{0} \rangle \\
        &=0.
    \end{align*}

        \item For $i>1$, the $i${th} row of $G_1$ with the second row of $H_1$:
    \begin{align*}
        \langle   (0,0, c_{i,j}| 0,0, d_{i,j}), (0,1,\textbf{0}|1,0,\textbf{0}) \rangle_s & =0 \cdot 1+0 \cdot 0+ \langle c_{i,j}, \textbf{0} \rangle\\
        &\hspace{0.4cm}+ 0 \cdot 0+0 \cdot 1+ \langle d_{i,j}, \textbf{0} \rangle \\
        &=0.
    \end{align*}

    \item For $i>1$, the $i${th} row of $G_1$ with the third row of $H_1$:
    \begin{align*}
        \langle   (0,0, c_{i,j}| 0,0, d_{i,j}), (1,0,u|1,0,v) \rangle_s & =0 \cdot 1+0 \cdot 0+ \langle c_{i,j}, v \rangle\\
        &\hspace{0.4cm}+ 0 \cdot 1+0 \cdot 0+ \langle d_{i,j}, u \rangle \\
        & = \langle (c_{i,j}| d_{i,j}), (u|v) \rangle_s \\
        &=0.~~~~~(\because~~(u|v) \in (C_{Res})^{\perp_S})
    \end{align*}

   \item  For $i>1$ and $j>3$, the $i${th} row of $G_1$ with the $j${th} row of $H_1$:
    \begin{align*}
        \langle   (0,0, c_{i,j}| 0,0, d_{i,j}), (0,0,s_{j,j'}|0,0,t_{j,j'}) \rangle_s & =0 \cdot 0+0 \cdot 0+ \langle c_{i,j}, t_{j,j'} \rangle\\
        &\hspace{0.4cm}+ 0 \cdot 0+0 \cdot 0+ \langle d_{i,j}, s_{j,j'} \rangle \\
        & = \langle (c_{i,j}| d_{i,j}), (s_{j,j'}| d_{j,j'}) \rangle_s \\
        &=0.
    \end{align*}
\end{itemize}
Since all the above considered symplectic inner products are zero, we have $G_1 \Omega_{2n+4} H_1^T=0$. On the other hand, the first three rows of $H_1$ are independent of all its other rows. This implies that $\rank(H_1)=3+(2n-k)$. Further, 
$$\rank(H_1)=3+(2n-k)=(2n+4)-(k+1)=\dim((D_{Res})^{\perp_S}).$$
Therefore, $H_1$ is a parity-check matrix of $D_{Res}$. Thus, $H'$ is a parity-check matrix of $D$.
\end{enumerate}
\end{proof}

To support our build-up construction methods, we now construct some free $E$-linear codes with larger lengths and symplectic hull-ranks from a free $E$-linear code with a smaller length and symplectic hull-rank. \par
In the following example, a free $E$-linear $[6,2]$-code with symplectic hull-rank $2$ is used to construct free $E$-linear $[8,3]$ and $[10,3]$-codes with symplectic hull-rank $3$.

\begin{example}
    Let $C$ be a free $E$-linear code generated by the matrix 
    $$M=\begin{pmatrix}
        \kappa & 0 & 0 & 0 & 0 & \kappa \\
        0 & \kappa & 0 & 0 & 0 & \kappa
    \end{pmatrix}.$$
     We checked by MAGMA \cite{Magma} that the free $E$-linear code $C$ has the symplectic hull-rank $2$. The residue code $C_{Res}$ is generated by 
    $$G=\begin{pmatrix}
        1 & 0 & 0 & 0 & 0 & 1 \\
        0 & 1 & 0 & 0 & 0 & 1
    \end{pmatrix}.$$
    If we choose $x=(1,1,0|0,0,1)$, then 
 \begin{align*}
     \langle x,r_1 \rangle_s & = \langle (1,1,0|0,0,1), (1,0,0|0,0,1) \rangle_s \\
     &= \langle (1,1,0), (0,0,1) \rangle + \langle (0,0,1), (1,0,0) \rangle \\
     & =0,
 \end{align*}
 and 
    \begin{align*}
     \langle x,r_2 \rangle_s & = \langle (1,1,0|0,0,1), (0,1,0|0,0,1) \rangle_s \\
     &= \langle (1,1,0), (0,0,1) \rangle + \langle (0,0,1), (0,1,0) \rangle \\
     & =0.
 \end{align*}
 Now, consider the free $E$-linear codes $D$ and $D'$ with respective generator matrices given by 

$$G_1=\begin{pmatrix}
     \kappa & \kappa & \kappa & 0 & \kappa & 0 & 0 & \kappa \\
     0 &  \kappa & 0 & 0 &0 & 0 & 0 & \kappa \\
     0&   0 & \kappa & 0 & 0&0 & 0 & \kappa
\end{pmatrix}~~~\text{and}~~~G_2=\begin{pmatrix}
     \kappa & \kappa & \kappa & \kappa & 0 & \kappa & \kappa & 0 & 0 & \kappa \\
    0 & 0 &  \kappa & 0 & 0 &0 &0 & 0 & 0 & \kappa \\
    0& 0&   0 & \kappa & 0 & 0&0 & 0 & 0& \kappa
\end{pmatrix}.$$
 Then, by \textbf{Construction I}, $D$ is a free $E$-linear $[8,3]$-code with symplectic hull-rank $3$, and by \textbf{Construction II}, $D'$ is a free $E$-linear $[10,3]$-code with symplectic hull-rank $3$.
\end{example}

The next example constructs free $E$-linear $[10,4]$ and $[12,4]$-codes with symplectic hull-rank $4$ using a free $E$-linear $[8,3]$-code  with symplectic hull-rank $3$.

\begin{example}
   Consider a free $E$-linear code $C$ generated by 
    $$M=\begin{pmatrix}
        \kappa & 0 & 0 & \kappa & 0 & \kappa & \kappa & 0 \\
        0 & \kappa & 0 & \kappa & \kappa & \kappa & 0 & 0 \\
        0 & 0 & \kappa & \kappa & 0 & \kappa & 0 & \kappa
    \end{pmatrix}.$$
    We checked by MAGMA \cite{Magma} that the symplectic hull-rank of $C$ is $3$. Here, if we choose $x=(0,0,1,0| 0,1,1,1)$, then $\langle x, r_i \rangle_s=0$ for $i=1,2,3$. Now, let $D$ and $D'$ be free $E$-linear codes with respective generator matrices
 \setcounter{MaxMatrixCols}{12}
 $$G_1=\begin{pmatrix}
        \kappa & 0 & 0 & \kappa & 0 & \kappa & 0 & \kappa & \kappa & \kappa  \\
       0 & \kappa & 0 & 0 & \kappa & 0 & 0 & \kappa & \kappa & 0 \\
        0 & 0 & \kappa & 0 & \kappa & 0 &  \kappa & \kappa & 0 & 0 \\
        0 & 0 & 0 & \kappa & \kappa & 0 & 0 & \kappa & 0 & \kappa
    \end{pmatrix}$$
    and 
    $$G_2=\begin{pmatrix}
     \kappa & \kappa & 0 & 0 & \kappa & 0 & \kappa & \kappa & 0 & \kappa & \kappa & \kappa \\
     0 & 0 & \kappa & 0 & 0 & \kappa & 0 & 0 & 0 & \kappa & \kappa & 0 \\
       0 & 0 & 0 & \kappa & 0 & \kappa & 0 & 0 &  \kappa & \kappa & 0 & 0 \\
       0 & 0 & 0 & 0 & \kappa & \kappa & 0 & 0  & 0 & \kappa & 0 & \kappa
    \end{pmatrix}.$$
Then, by \textbf{Construction I}, $D$ is a free $E$-linear $[10,4]$-code with symplectic hull-rank $4$, and by \textbf{Construction II}, $D'$ is a free $E$-linear $[12,4]$-code with symplectic hull-rank $4$.
   
\end{example}

\section{Hull-variation problem}
  This section investigates the permutation equivalence of free $E$-linear codes and classifies the optimal free $E$-linear codes with a fixed symplectic hull-rank. We also study the symplectic hull-variation problem for the free $E$-linear codes.\par

Let $\alpha$ be a permutation. We observe that $\langle \alpha(w), \alpha(z) \rangle=\langle w, z \rangle$ for all $w,z \in E^{2n}$. But this may not be true for the symplectic inner product, as shown in the following example.

\begin{example}
    Fix $n=2$, and let $w=(\kappa,0|0,0)$ and $z=(0,\kappa|0,0)$ be two vectors in $E^4$. Then, their symplectic inner product is given by 
    \begin{align*}
    \langle w,z \rangle_s &=\langle (\kappa,0|0,0), (0,\kappa|0,0) \rangle_s \\
    &=\langle (\kappa,0),(0,0) \rangle+\langle (0,0), (0,\kappa) \rangle \\
    &=0.
    \end{align*}
    Now, let $\alpha=(2,3)$ be a permutation swapping the second and the third coordinate, fixing others. Then, we have
    $$\alpha(w)=(\kappa,0|0,0)~~~\text{and}~~~\alpha(z)=(0,0|\kappa,0).$$
    Hence, their symplectic inner product is given by 
    \begin{align*}
        \langle \alpha(w), \alpha(z) \rangle_s &=\langle (\kappa,0|0,0), (0,0|\kappa,0) \rangle_s \\
        &=\langle (\kappa,0),(\kappa,0) \rangle+\langle (0,0), (0,0) \rangle \\
        &=\kappa.
    \end{align*}
    Thus, $\langle \alpha(w), \alpha(z) \rangle_s \neq\langle w, z \rangle_s$.
\end{example}

Now, let $C$ be a free $E$-linear $[2n,k]$-code and $\alpha$ be a permutation. It is clear that $\alpha(C^{\perp})=\alpha(C)^\perp$ where $C^\perp$ is the Euclidean dual of the code $C$. The following example shows that this need not be true for the symplectic dual.

\begin{example}
    Let $C$ be a free $E$-linear $[4,1]$-code generated by $G=\begin{pmatrix}
        \kappa & 0 & 0 & 0
    \end{pmatrix}$. Then its symplectic dual $C^{\perp_S}$ is generated by the matrix $$H=\begin{pmatrix}
        \kappa & 0 & 0 & 0 \\
         0& \kappa & 0 & 0 \\
        0 & 0 & 0 & \kappa
    \end{pmatrix}.$$
    Now, fix a permutation $\alpha=(2,3)$. Then the corresponding generator matrices of $\alpha(C)$ and $\alpha(C^{\perp_S})$ are 

    $$\alpha(G)=\begin{pmatrix}
         \kappa & 0 & 0 & 0
    \end{pmatrix} ~~~~\text{and}~~~~\alpha(H)=\begin{pmatrix}
         \kappa & 0 & 0 & 0 \\
         0&  0 &\kappa & 0 \\
        0 & 0 & 0 & \kappa
    \end{pmatrix}.$$
    Further, $\alpha(C)=C$ implies that $H$ is a generator matrix of $(\alpha(C))^{\perp_S}$. We see that $(0,\kappa,0,0) \in (\alpha(C))^{\perp_S}$ but $(0,\kappa,0,0) \notin \alpha(C^{\perp_S})$. Thus, $\alpha(C^{\perp_S}) \neq (\alpha(C))^{\perp_S}$.
\end{example}

\begin{definition} For two codes $C$ and $D$, if $D=CP$ for a permutation matrix $P$, they are called permutation-equivalent.
 \end{definition}
 
The following result from \cite{Alah23} investigates the permutation equivalence of free $E$-codes. 

\begin{theorem}( \cite{Alah23}, Theorem $16$)\label{Thm5}
  Let  $C$ and $D$  be two free $E$-linear codes. Then they are permutation-equivalent if and only if their residue codes are permutation-equivalent. 
\end{theorem}

We now investigate the symplectic hull-variation problem for the free $E$-linear codes.

\begin{theorem}
    Let $C$ and $D$ be two permutation-equivalent free $E$-linear codes of length $2n$. If $P$ is a permutation matrix such that $D=CP$, then their symplectic hull-ranks remain invariant if $P \Omega_{2n}P^T=\Omega_{2n}$.
\end{theorem}
\begin{proof}
Let $C$ and $D$ be two permutation-equivalent free $E$-linear codes and $D=CP$ for a permutation matrix $P$. Then, by Theorem \ref{Thm5}, their residue codes $C_{Res}$ and $D_{Res}$ are also permutation-equivalent and $D_{Res}=C_{Res}P$. Now, assume that $G_1$ and $G_2$ are generator matrices of $C_{Res}$ and $D_{Res}$, respectively. Without loss of generality, we may assume that $G_2=G_1P$. Then, we have
$$\rank(G_2\Omega_{2n}G_2^T)=\rank(G_1P\Omega_{2n}P^TG_1^T)=\rank(G_1\Omega_{2n}G_1^T).$$
 Therefore, by Theorem \ref{Thm4}, the symplectic hull-dimensions of $C_{Res}$ and $D_{Res}$ are equal. Thus, by Theorem \ref{thm3}, the symplectic hull-ranks of $C$ and $D$ are equal. This completes the proof.

\end{proof}

For two permutation-equivalent free codes $C$ and $D$ over $E$, if $D=CP$ but $P \Omega_{2n}P^T \neq \Omega_{2n}$, then their symplectic hull-ranks may differ as shown in the example given below.

\begin{example}
    Let $C$ be a free $E$-linear code of length $4$ generated by the matrix 
    $$G=\begin{pmatrix}
        \kappa & 0 & 0 & 0 \\
        0 & \kappa & 0 & 0
    \end{pmatrix}.$$
    Then, its symplectic dual $C^{\perp_S}$ is generated by  
$$H=\begin{pmatrix}
        \kappa & 0 & 0 & 0 \\
        0 & \kappa & 0 & 0
    \end{pmatrix}.$$
    From the above, we conclude that $C=C^{\perp_S}$. Consequently, $\SHull(C)=C$ and the symplectic hull-rank of $C$ is $2$. Choose a permutation matrix $P$ given by
    $$P=\begin{pmatrix}
        1 & 0 & 0 & 0 \\
        0 & 0 & 1 & 0 \\
        0 & 1 & 0 & 0 \\
        0 & 0 & 0 & 1
    \end{pmatrix}.$$
   Next, consider the free $E$-linear code $D$ given by $D=CP$. Then, a generator matrix of $D$ can be given by 
    $$G'=\begin{pmatrix}
        \kappa & 0 & 0 & 0 \\
        0 & 0 & \kappa & 0
    \end{pmatrix}.$$
    It is easy to see that a generator matrix of $D^{\perp_S}$ can be given by

    $$H'=\begin{pmatrix}
        0 & \kappa & 0 & 0 \\
        0 & 0 & 0 & \kappa
    \end{pmatrix}.$$
    Clearly, $D \cap D^{\perp_S}=\{0\}$. Hence, the symplectic hull of the code $D$ is trivial, and so its symplectic hull-rank is zero. Thus, the symplectic hull-ranks of the permutation-equivalent codes $C$ and $D$ are distinct. Note that 
  $$P\Omega_4P^T=\begin{pmatrix}
      0 & 1 & 0 & 0 \\
      1 & 0 & 0 & 0\\
      0 & 0 & 0 & 1\\
      0 & 0 &1 & 0
  \end{pmatrix} \neq \Omega_4.$$
    
\end{example}

In what follows, we utilize the notion of optimal codes in a way consistent with the approach in \cite{Li}.
\begin{definition}\label{defn2}
    A free $[2n,k]$-code $C$ over $E$ with symplectic hull-rank $l=i$ is said to be $l_i$-optimal if the symplectic distance of the code $C$ is maximum among all the free $E$-linear $[2n,k]$-codes with symplectic hull-rank $i$.
\end{definition}

We now present the classification of optimal free $E$-linear codes with a prescribed symplectic hull-rank. First, we apply Theorem $3.5$ of \cite{Kushwaha2} to obtain all the free $E$-linear codes for lengths up to $4$. Then, Theorem \ref{Thm5} is used to determine all permutation-inequivalent free $E$-linear codes. We then compute the symplectic hull-rank of each code using Theorem \ref{thm3}. Finally, the optimality of these codes is assessed in accordance with Definition \ref{defn2} and the resulting optimal codes are listed below. All computations in this work have been performed using MAGMA \cite{Magma}.
\begin{itemize}
    \item \textbf{Inequivalent free $E$-linear $l_0$-optimal codes}\\\\
The generator matrices of the inequivalent free $E$-linear optimal codes with symplectic hull-rank $0$ are given below:\par
$ \begin{pmatrix}
    \kappa & 0 \\
    0 & \kappa
\end{pmatrix}$, 
$\begin{pmatrix}
    \kappa & 0 & 0 & 0 \\
    0 & 0 & \kappa & 0
\end{pmatrix}$,
$\begin{pmatrix}
    \kappa & 0 & 0 & \kappa \\
    0 & \kappa & 0 & \kappa
\end{pmatrix}$, 
$\begin{pmatrix}
    \kappa & 0 & 0 & \kappa \\
    0 & \kappa & 0 & 0
\end{pmatrix}$,
$\begin{pmatrix}
    \kappa & 0 & \kappa & \kappa \\
    0 & \kappa & 0 & \kappa
\end{pmatrix}$,
$\begin{pmatrix}
    \kappa & 0 & \kappa & \kappa \\
    0 & \kappa & 0 & 0
\end{pmatrix}$, and 
$\begin{pmatrix}
    \kappa & 0 & 0 & 0 \\
    0 & \kappa & 0 & 0 \\
    0 & 0 & \kappa & 0 \\
    0 & 0 & 0 & \kappa
\end{pmatrix}$.\\\\

\item \textbf{Inequivalent free $E$-linear $l_1$-optimal codes} \\\\
Inequivalent free $E$-linear optimal codes with symplectic hull-rank $1$ are given below:\par
 $\begin{pmatrix}
     \kappa & 0
 \end{pmatrix}$,
$\begin{pmatrix}
     \kappa & \kappa
 \end{pmatrix}$,
 $\begin{pmatrix}
     \kappa & 0 & \kappa & \kappa
 \end{pmatrix}$,
 $\begin{pmatrix}
     \kappa & 0 & 0 & \kappa
 \end{pmatrix}$,
 $\begin{pmatrix}
     \kappa & \kappa & \kappa & \kappa
 \end{pmatrix}$,
 $\begin{pmatrix}
     \kappa & \kappa & 0 & 0 \\
     0 & 0 & \kappa & 0 \\
     0 & 0 & 0 & \kappa
 \end{pmatrix}$,
 $\begin{pmatrix}
     \kappa & 0 & 0 & 0 \\
     0 & \kappa & 0 & 0 \\
     0 & 0 & 0 & \kappa
 \end{pmatrix}$, 
  $\begin{pmatrix}
     \kappa & 0 & 0 & \kappa \\
     0 & \kappa & 0 & 0 \\
     0 & 0 & \kappa & \kappa
 \end{pmatrix}$, and 
 $\begin{pmatrix}
     \kappa & 0 & 0 & \kappa \\
     0 & \kappa & 0 & \kappa \\
     0 & 0 & \kappa & \kappa
 \end{pmatrix}$. \\\\

\item \textbf{Inequivalent free $E$-linear $l_2$-optimal codes}\\\\
The free $E$-linear codes with generator matrices
 $\begin{pmatrix}
     \kappa & 0 & 0 & \kappa \\
     0 & \kappa & \kappa & \kappa
 \end{pmatrix}$ and $\begin{pmatrix}
      \kappa & 0 & 0 & \kappa \\
     0 & \kappa & \kappa & 0
 \end{pmatrix}$ are the only inequivalent free $E$-linear optimal codes with symplectic hull-rank $2$.

\end{itemize}
\section{Conclusion}
In this manuscript, we have investigated the three symplectic hulls of an $E$-linear code. Initially, we have identified the residue and torsion codes of the three symplectic hulls and determined the generator matrix of the two-sided symplectic hull of a free $E$-linear code. Additionally, the symplectic hull-rank of a free $E$-linear code is calculated. Then, we studied the symplectic hull of the sum of the two free $E$-linear codes and also produced the condition for the sum of the two free $E$-linear codes to be a symplectic LCD code. Subsequently, we proposed two build-up construction techniques that extend a free $E$-linear code with a smaller length and symplectic hull-rank to one of larger length and symplectic hull-rank. We supported our build-up methods through concrete examples of code. Furthermore, we discussed the permutation equivalence and then investigated the symplectic hull-variation problem for the free $E$-linear codes. Finally, the classification of free $E$-linear optimal codes has been presented for lengths up to $4$. A potential direction for future research can be the extension of the present study to other non-unital rings appearing in Fine’s classification \cite{Fine93}.

\section*{Acknowledgement}
The first author thanks the Indian Institute of Technology Patna for providing research facilities and sincerely acknowledges the financial support from the Council of Scientific \& Industrial Research, Govt. of India (under grant No. 09/1023(16098)/2022-EMR-I).
\section*{Declarations}
\textbf{Competing interests}: There is no competing interest among authors associated with this work.\\
\textbf{Data Availability}: All data used to derive the results of this study are included in the manuscript. If necessary, additional information can be obtained from the corresponding author. \\
\textbf{Use of AI tools}: Artificial Intelligence (AI) tools were not used in writing or preparation of this work.

\end{document}